\documentclass[conference,letterpaper]{IEEEtran}

\usepackage[margin=0.55in]{geometry}
\usepackage[utf8]{inputenc} 
\usepackage[T1]{fontenc}
\usepackage{ifthen}
\usepackage{cite}
\usepackage[cmex10]{amsmath} 
\usepackage{amsthm,amsfonts,amssymb}
\usepackage{graphicx}
\usepackage{subfig}
\usepackage{blkarray}
\usepackage[mathscr]{euscript}
\usepackage{enumitem}
\usepackage{blkarray}
\newtheorem{theorem}{Theorem}[section]

\theoremstyle{remark}
\newtheorem*{remark}{Remark}
\theoremstyle{definition}
\newtheorem{definition}{Definition}[section]
\usepackage{url}
\usepackage{arydshln}
\usepackage{mathtools}
\usepackage{dsfont}
\usepackage{kbordermatrix}

\begin{document}
\title{Computable Lower Bounds for Capacities of Input-Driven Finite-State Channels} 


\author{%
  \IEEEauthorblockN{V.~Arvind Rameshwar}
  \IEEEauthorblockA{Indian Institute of Science, Bengaluru\\
                    Email: \texttt{vrameshwar@iisc.ac.in}}
  \and
  \IEEEauthorblockN{Navin Kashyap}
  \IEEEauthorblockA{Indian Institute of Science, Bengaluru\\
				  	Email: \texttt{nkashyap@iisc.ac.in}
}
}

\maketitle
\begin{abstract}
This paper studies the capacities of input-driven finite-state channels, i.e., channels whose current state is a time-invariant deterministic function of the previous state and the current input. We lower bound the capacity of such a channel using a dynamic programming formulation of a bound on the maximum reverse directed information rate. We show that the dynamic programming-based bounds can be simplified by solving the corresponding Bellman equation explicitly. In particular, we provide analytical lower bounds on the capacities of $(d,k)$-runlength-limited input-constrained binary symmetric and binary erasure channels. Furthermore, we provide a single-letter lower bound based on a class of input distributions with memory.
\end{abstract}

\section{Introduction}

Discrete finite-state channels (DFSCs or FSCs) are mathematical models for channels with finite memory that are applied to magnetic recording \cite{Magnetic}, and wireless communications \cite{Wireless}. The channel memory is encapsulated in the channel state, which takes values from a finite set. Although Shannon's single-letter expression exists for the capacity of discrete memoryless channels (DMCs) \cite{Shannon}, namely, $C_{DMC} = \max_{p(x)} I(X;Y)$, the capacity of FSCs, with the exception of special cases, is characterized only as a multi-letter expression.

In this paper, we derive lower bounds on the capacity of input-driven FSCs (Fig. \ref{fig1}), where the channel state evolves as a time-invariant function of the state at the previous time instant, and the current input. This class of channels includes the collection of DMCs with input constraints, which will be treated in some detail in our work. Existing work on lower bounding the capacity of FSCs includes the simulation-based approach in \cite{Arnold}, the generalized Blahut-Arimoto algorithm developed in \cite{GBAA}, and the stochastic approximation algorithm proposed in \cite{Han3}, all of which are numerical methods. Analytical lower bounds were derived by Zehavi and Wolf \cite{ZW88} for binary symmetric channels with a $(d,k)$-runlength-limited (RLL) constraint --- see Definition~\ref{def:RLL} --- at the input. Later works gave capacity lower bounds for input-constrained binary symmetric and binary erasure channels in the asymptotic (very low or very high noise) regimes \cite{Han1}, \cite{Han2}, \cite{Ordentlich}. Our work applies to a larger class of channels, and provides bounds for all values of the channel parameters. 

The key idea is the lower bounding of the $N$-letter mutual information between the channel inputs and the outputs, $I(X^N;Y^N)$, by the reverse directed information \cite{Kramer}, $I(Y^N\rightarrow X^N)$. We show that the derived lower bound on the capacity can be formulated as an infinite-horizon average-reward dynamic programming (DP) problem, and is hence computable. 

Further, we show that the lower bound derived can, in turn, be bounded below by a single-letter expression obtained using input distributions on a directed ``$\mathscr{V}$-graph''. Our treatment here is entirely analogous to the single-letter lower bounding technique introduced in \cite{Single}, which uses input distributions on a ``Q-graph'' that is obtained by a recursive quantization of channel outputs on a directed graph.

We then apply the DP-based lower bound to the class of input-constrained binary symmetric channels (BSCs) and binary erasure channels (BECs). DP problems are typically handled by solving the corresponding Bellman equations. We consider the $(d,k)$-RLL input-constrained (RIC) BSC and BEC, and explicitly solve the Bellman equations for the DP-based lower bounds for each of these channels.  Interestingly, our techniques recover the lower bounds given in \cite{ZW88}, for the $(d,k)$-RIC BSC, for $k<\infty$. For the $(1,\infty)$-RIC BSC and BEC, the analytical lower bounds thus found compare favourably with asymptotic lower bounds given in \cite{Han1}, \cite{Han2}, and extend to all values of the channel parameters.

\begin{figure}[!t]
  \centering
  \includegraphics[width=0.48\textwidth]{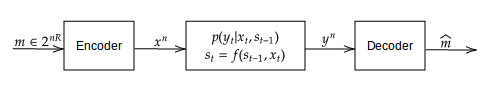}
   \caption{System model of an input-driven FSC.}
   \label{fig1}
\end{figure}

The paper is organized as follows. Section~\ref{sec:notation} contains the necessary information-theoretic preliminaries, Section~\ref{sec:results} states the main results, and Section~\ref{sec:DP} explains the DP formulation, which is used in Section~\ref{sec:examples} to derive explicit lower bounds on the capacity of the $(d,k)$-RIC BEC and BSC. The proofs of our main results are presented in Section~\ref{sec:proofs}. Some concluding remarks are made in Section~\ref{sec:conc}.

\section{Notation and Preliminaries}\label{sec:notation}

In this section, we introduce notation, the channel model, the lower bound on mutual information rate, and the definitions of $\mathscr{V}$-graphs and $(\mathscr{S},\mathscr{V})$-graphs.

\subsection{Notation}

In what follows, random variables will be denoted by capital letters, and their realizations by lower-case letters, e.g., $X$ and $x$, respectively. Calligraphic letters, e.g., $\mathscr{X}$, denote sets. The notations $X^{N}$ and $x^N$ denote the random vector $(X_1,\ldots,X_N)$ and the realization $(x_1,\ldots,x_N)$, respectively. Further, $P(x), P(y)$ and $P(y|x)$ are used to denote the probabilities $P_X(x), P_Y(y)$ and $P_{Y|X}(y|x)$, respectively. As is usual, the notations $H(X)$ and $I(X;Y)$ stand for the entropy of the random variable $X$, and the mutual information between the random variables $X$ and $Y$, respectively, and $h_b(p)$ and $H_{ter}(p,q)$ are the binary and ternary entropy functions, respectively. Finally, for a real number $\alpha \in [0,1]$, we define $\bar{\alpha}=1-\alpha$.

\subsection{Channel Model}\label{SS2}

Consider the family of channels with input $x_t\in \mathscr{X}$, output $y_t\in \mathscr{Y}$, and state $s_t\in \mathscr{S}$ at time $t$, with $\mathscr{X},\mathscr{Y},\mathscr{S}$ having finite cardinalities. At each time $t$, an FSC obeys $P(s_t,y_t|x^t,s^{t-1},y^{t-1})=P(s_t,y_t|x_t,s_{t-1})$. An \emph{input-driven} FSC has the additional property that there exists a time-invariant function, $f: \mathscr{S}\times \mathscr{X}\rightarrow \mathscr{S}$, such that $s_t=f(s_{t-1},x_t)$. 

%
%

We shall assume throughout that the initial state $s_0$ of the FSC is known to the encoder and decoder. While this assumption can be removed through a suitable notion of channel indecomposability \cite[(4.6.26)]{Gallager}, we retain the assumption as it is realistic in the context of input-constrained DMCs, which is the main application of interest to us. The following theorem gives an expression for the capacity of such FSCs.

\begin{theorem}[\cite{Gallager}, Ch.~$4.6$]
	The capacity of an FSC with known initial state $s_0$ is given by
	\begin{equation*}
	C = \lim_{N\rightarrow \infty} \max_{P(x^N | s_0)} \frac{1}{N}I(X^N;Y^N|s_0).
	\end{equation*}
	\label{thm:FSC_cap}
\end{theorem}

From here on, we will drop the explicit conditioning on $s_0$ in our notation, including it only when there is need.

We also introduce below the definition of a connected FSC, which we shall use in Theorem~\ref{Thm5}.

\begin{definition}
\label{def:connected}
An FSC is \emph{connected} if, for each $s \in \mathscr{S}$, there is an input distribution $\{P(x_n|s_{n-1})\}_{n \ge 1}$ and an integer $N$ such that $\sum_{n=1}^N P_{S_n|S_0} (s|s') > 0$, for all $s' \in \mathscr{S}$.
\end{definition}

\subsection{Directed Information}

We recall the definition of directed information, introduced by Massey \cite{Massey}:

\begin{definition}
	The (forward) \emph{directed information} is given by
	\begin{align*}
	I(X^N\rightarrow Y^N)&:=\sum_{t=1}^{N} I(X^t;Y_t|Y^{t-1})\\
	&=\mathbb{E}\left[\log_2\left(\frac{P(Y^N||X^N)}{P(Y^N)}\right)\right],
	\end{align*}
	where $P(y^N||x^N):=\prod_{t=1}^{N}P(y_t|x^{t},y^{t-1})$ is the \emph{causal conditioning distribution}.
\end{definition}
Analogously, we define the \emph{reverse} directed information by
\begin{equation*}
	I(Y^N\rightarrow X^N):=\sum_{t=1}^{N} I(Y^t;X_t|X^{t-1}).
	\end{equation*}
In addition, we make use of the following definition:
\begin{align*}
I(X^{N-1}\rightarrow Y^N)&:=\sum_{t=1}^{N} I(X^{t-1};Y_t|Y^{t-1})\\
&=\mathbb{E}\left[\log_2\left(\frac{P(Y^N||X^{N-1})}{P(Y^N)}\right)\right],
\end{align*}
where $P(y^N||x^{N-1}):=\prod_{t=1}^{N}P(y_t|x^{t-1},y^{t-1})$.


The following conservation law for information is well-known \cite[Prop.~2]{Massey_isit05}:
\begin{equation}
	I(X^N;Y^N)=I(Y^N\rightarrow X^N)+I(X^{N-1}\rightarrow Y^N).
	\label{eq:conserv}
\end{equation}
In particular, mutual information is bounded below by reverse directed information, i.e.,
\begin{equation}
\label{LB}
I(X^N;Y^N)\geq I(Y^N\rightarrow X^N).
\end{equation}

\subsection{The $\mathscr{V}$-graph and $(\mathscr{S},\mathscr{V})$-graph}
Similar to the $Q$-graph and the $(\mathscr{S},Q)$-graph of \cite{Single}, we introduce the following definitions:

\begin{definition}
	A $\mathscr{V}$-graph is a finite irreducible labelled directed graph on a vertex set $\mathscr{V}$, with the property that each $v\in \mathscr{V}$ has at most $|\mathscr{X}|$ outgoing edges, each labelled by a unique $x\in \mathscr{X}$.
\end{definition}
	
	Thus, there exists a function $\Phi: \mathscr{V} \times \mathscr{X} \rightarrow \mathscr{V}$, such that $\Phi(v,x)=v'$ iff there is an edge $v \stackrel{x}{\longrightarrow} v'$ in the $\mathscr{V}$-graph.
We arbitrarily label one vertex of the $\mathscr{V}$-graph as $v_0$. For any positive integer $n$, there is a one-to-one correspondence between sequences in $(x_1,x_2,\ldots,x_n) \in \mathscr{X}^n$ and directed paths in the $\mathscr{V}$-graph  starting from $v_0$: $v_0 \stackrel{x_1}{\longrightarrow} v_1 \stackrel{x_2}{\longrightarrow}
\cdots \stackrel{x_n}{\longrightarrow} v_n$. Fig. \ref{figV} depicts an example of a $\mathscr{V}$-graph.

\begin{figure}[htbp]
  \centering
  \includegraphics[width=0.45\textwidth]{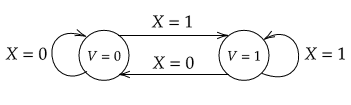}
   \caption{A $\mathscr{V}$-graph where each node represents the last channel input, where $\mathscr{X}=\{0,1\}$.}
   \label{figV}
\end{figure}

\begin{definition}
	Given an input-driven FSC specified by $\{P(y|x,s)\}$ and $s_t=f(s_{t-1},x_t)$, and a $\mathscr{V}$-graph with vertex set $\mathscr{V}$, the $(\mathscr{S},\mathscr{V})$-graph is defined to be a directed graph on the vertex set $\mathscr{S}\times \mathscr{V}$, with edges $(s,v) \xrightarrow{(x,y)} (s^{\prime},v^{\prime})$ if and only if $P(y|x,s)>0$, $s^{\prime}=f(s,x)$, and $v^{\prime}=\Phi(v,x)$.
\end{definition}

Now, given an input distribution $\{Q(x|s,v)\}$ defined for each $(s,v)$ in the $(\mathscr{S},\mathscr{V})$-graph, we have a Markov chain on $\mathscr{S}\times \mathscr{V}$, where the transition probability associated with any edge $(x,y)$ emanating from $(s,v)\in \mathscr{S}\times \mathscr{V}$ is $P(y|x,s)Q(x|s,v)$. Let $\mathscr{G}(\{Q(x|s,v)\})$ be the subgraph remaining after discarding edges of zero probability. We then define
\begin{align*}
\mathscr{Q} \triangleq \bigl\lbrace \{Q(x | s,v)\}: \mathscr{G}& (\{Q(x|s,v)\})\text{ has a single} \\ 
& \text{closed communicating class}\bigr\rbrace.
\end{align*}
Given an irreducible $\mathscr{V}$-graph, an input distribution $\{Q(x|s,v)\} \in \mathscr{Q}$ is said to be \emph{aperiodic}, if the corresponding graph, $\mathscr{G}(\{Q(x|s,v)\})$, is aperiodic. For such distributions, the Markov chain on $\mathscr{S}\times \mathscr{V}$ has a unique stationary distribution $\pi(s,v)$.

\section{Main Results}\label{sec:results}
We shall now restrict attention to input-driven FSCs, defined in Section \ref{SS2}. We assume that the initial channel state, $s_0$, is chosen deterministically, and is known to both the encoder and the decoder. We present a lower bound on the capacity of indecomposable input-driven FSCs.

\begin{theorem}\label{Thm3}
    The capacity of an input-driven FSC with known initial state is bounded below as:
    \begin{align*}
        C&\geq \lim_{N\rightarrow \infty} \max_{\{P(x_t|s_{t-1})\}_{t= 1}^{N}}\frac{1}{N}\sum_{t=1}^{N}I(X_t,S_{t-1};Y_t \mid X^{t-1})\\
        &=\sup_{\{P(x_t|s_{t-1})\}_{t \geq 1}} \liminf_{N\rightarrow \infty} \frac{1}{N}\sum_{t=1}^{N}I(X_t,S_{t-1};Y_t \mid X^{t-1}),
    \end{align*}
    where the conditional distribution $P_{X_t,S_{t-1},Y_t|X^{t-1}}(x_t,s_{t-1},y_{t}|x^{t-1})=\beta_{t-1}(s_{t-1})\times P(x_t|s_{t-1})P(y_t|x_t,s_{t-1})$, with $\beta_{t-1}(s_{t-1})=P(s_{t-1}|x^{t-1})$.
\end{theorem}

\begin{remark}\label{rem1}
	When $Y^N=X^N$, it can be seen that:
	\begin{align*}
	    \sum_{t=1}^{N}I(X_t,S_{t-1};Y_t \mid X^{t-1})&=\sum_{t=1}^{N} \left[H(X_t|X^{t-1})- H(X_t|X_t)\right]\\
	    &=H(X^N).
	\end{align*}
	Hence, the lower bound is tight, and is equal to the maximum entropy rate of the input process.
\end{remark}
The proof of Theorem \ref{Thm3} is given in section \ref{sec:proofs}. The computability of the lower bound follows from Theorem \ref{Thm4}.

\begin{theorem}\label{Thm4}
      The lower bound expression in Theorem \ref{Thm3} can be formulated as an infinite-horizon average-reward dynamic program, where the DP state is the probability vector $\beta_{t-1} = \bigl(P(s_{t-1}|x^{t-1}): s_{t-1}\in \mathscr{S}\bigr)$, the action is the stochastic matrix $[P(x_t|s_{t-1})]$, and the disturbance is the channel input, $x_t$.
\end{theorem}
We also propose the following alternative single-letter lower bound, when the channel also obeys the connectedness property defined in Section \ref{SS2}.

\begin{theorem}\label{Thm5}
	For a connected input-driven FSC, given a $\mathscr{V}$-graph on the inputs,
	\begin{equation*}
		C\geq I_{{Q}}(X;Y|S,V),
	\end{equation*}
where $\{Q(x|s,v)\} \in \mathscr{Q}$ is an aperiodic distribution that induces the stationary distribution $\pi(s,v)$ on the corresponding $(\mathscr{S},\mathscr{V})$-graph, and the random variables $X,Y,S,V$ are jointly distributed as $P_{X,Y,S,V}(x,y,s,v) = \pi(s,v) Q(x|s,v) P(y|x,s)$.
	
\end{theorem}
The proof of Theorem \ref{Thm5} is presented in Section \ref{sec:proofs}. Theorems \ref{Thm4} and \ref{Thm5} provide powerful techniques to arrive at analytical lower bounds. We then applied Theorem \ref{Thm3} to input-constrained memoryless channels, and obtained the following lower bounds:
\begin{itemize}
    \item The capacity of the $(d,\infty)$-RIC BSC($p$) satisfies
    \begin{equation*}
        C\geq \max_{a\in [0,1]}\frac{h_b(ap+\bar{a}\bar{p})-h_b(p)}{ad+1}.
    \end{equation*}This result holds for all $p\in [0,1]$, and, for $d=1$, numerical evaluations indicate that the DP bound is close to the asymptotic bounds of \cite{Han1} (as $p\rightarrow 0$), and of \cite{Ordentlich} (as $p\rightarrow 0.5$).
    \item The capacity of the $(d,k)$-RIC BSC($p$) obeys 
    \begin{equation*}
        C\geq \max_{a_d,\ldots,a_{k-1}}\frac{\sum\limits_{i=d}^{k-1}(h_b(a_ip+\bar{a_i}\bar{p})-h_b(p))\prod\limits_{j=d}^{i-1}(1-a_j)}{d+1+\sum\limits_{i=d}^{k-1}\prod\limits_{j=d}^{i}(1-a_j)},
    \end{equation*}
    where $a_d,\ldots,a_{k-1}\in [0,1]$. These lower bounds hold for arbitrary $0\leq d<k<\infty$.
    \item For $0 \le d < k \le \infty$, the capacity of the $(d,k)$-RIC BEC($\epsilon$) satisfies $C\geq C_{d,k} \cdot\bar{\epsilon}$, where $C_{d,k}$ is the noiseless capacity of the $(d,k)$-RLL constraint. In particular, when $d=0$, the bound becomes tight as $k\rightarrow \infty$, and for $d=1$ and $k=\infty$, it extends the asymptotic results of \cite{Han2}.
\end{itemize}

\section{DP Formulation}\label{sec:DP}
In this section, we shall formulate the lower bound in Theorem \ref{Thm3} as a DP problem, thereby showing the validity of Theorem \ref{Thm4}. We also introduce the Bellman equation, that provides a sufficient condition for optimality of the reward.

\subsection{DP Problem}
An infinite-horizon average-reward DP is defined by the tuple $(\mathscr{Z},\mathscr{U},\mathscr{W},F,P_Z,P_w,g)$. We consider a discrete-time dynamic system evolving according to:
\begin{equation*}
    z_t=F(z_{t-1},u_t,w_t),\quad t=1,2,\ldots
\end{equation*}
Each state, $z_t$, takes values in a Borel space $\mathscr{Z}$, each action, $u_t$, in a compact subset $\mathscr{U}$ of a Borel space, and each disturbance, $w_t$, in a measurable space $\mathscr{W}$. The initial state, $z_0$, is drawn from $P_Z$, and the disturbance, $w_t$, from $P_w(\cdot|z_{t-1},u_t)$. At time $t$, the action $u_t$ is equal to $\mu_t(h_t)$, where $h_t:=(z_0,w^{t-1})$ is the history up to time $t$. A policy $\pi$ is defined as $\pi:=\{\mu_1,\mu_2,\ldots\}$.

The aim is to maximize the average reward, given a bounded reward function $g:\mathscr{Z}\times \mathscr{U}\rightarrow \mathbb{R}$. For a policy $\pi$, the average reward is:
\begin{equation}\label{eq:reward}
    \rho_{\pi}:=\liminf_{N\rightarrow \infty} \frac{1}{N}\mathbb{E}_{\pi}\left[\sum_{t=1}^{N}g(Z_{t-1},\mu_t(h_t))\right],
\end{equation}
where the subscript $\pi$ indicates that the actions are generated by $\pi=(\mu_1,\mu_2,\ldots)$. The optimal average reward is defined as $\rho^{*}:=\sup_{\pi} \rho_{\pi}$.

\subsection{Lower Bound of Theorem \ref{Thm3} as a DP Problem}\label{DP}
From the DP formulation of the average reward in equation \eqref{eq:reward} and from Theorem \ref{Thm3}, the DP state is chosen to be $z_{t-1}\stackrel{\Delta}{=}\beta_{t-1}=\bigl(P(s_{t-1}|x^{t-1}): s_{t-1}\in \mathscr{S}\bigr)$. The action, $u_t$, is the stochastic matrix $[P(x_t|s_{t-1})]$, and the disturbance, $w_t$, is the channel input, $x_t$, which takes values in $\mathscr{X}$. Further, we define the reward function,
\begin{equation*}
    g(\beta_{t-1},u_t)=I(X_t,S_{t-1};Y_t|x^{t-1}).
\end{equation*}
It is easy to see that the average of the reward function is indeed the lower bound of Theorem \ref{Thm3}. We now verify that the formulation above satisfies the properties of a DP problem. Firstly, we note that the conditional distribution
\begin{equation}\label{eq1}
    P(x_t,s_{t-1},y_t|x^{t-1})=z_{t-1}(s_{t-1})P(x_t|s_{t-1})P(y_t|x_t,s_{t-1}),
\end{equation}
which depends only on the previous state and the action. Therefore, the reward function, $g(\cdot)$ at time $t$, is a function of only $z_{t-1}$ and $u_t$. Secondly, it is easy to check that the disturbance distribution depends only on $z_{t-1}$ and $u_t$: 
\begin{align*}
    &P(w_t|w^{t-1},z^{t-1},u^t)\\&=P(x_t|x^{t-1},\beta^{t-1},u^t)\\
    &\stackrel{(a)}{=}\sum_{y_t,s_{t-1}} \beta_{t-1}(s_{t-1})P(x_t|s_{t-1},x^{t-1},\beta^{t-1},u^{t})P(y_t|x_t,s_{t-1})\\
    &\stackrel{(b)}{=}\sum_{y_t,s_{t-1}}\beta_{t-1}(s_{t-1})P(x_t|s_{t-1},\beta_{t-1},u_t)P(y_t|x_t,s_{t-1})\\
    &=P(x_t|\beta_{t-1},u_t)\\
    &=P(w_t|\beta_{t-1},u_t),
\end{align*}
where (a) follows from the fact that the value of $P(s_{t-1}|x^{t-1},\beta^{t-1},u^{t})$ is determined by $\beta_{t-1}(s_{t-1})$, and (b) follows from the fact that the distribution of $x_t$ depends only on the triplet $(s_{t-1},\beta_{t-1},u_t)$. Lastly, we need to show that there exists a deterministic function $F:\mathscr{Z}\times \mathscr{U}\times \mathscr{W}\rightarrow \mathscr{Z}$, such that $z_t=F(z_{t-1},u_t,w_t)$. But we know that
\begin{equation*}
    z_t(s_t) = P(s_t|x^t) = \sum_{s_{t-1}} P(s_{t-1}|x^{t-1})\mathds{1}\{f(s_{t-1},x_t)=s_t\}.
\end{equation*}
Clearly, the next DP state is a function of the previous state and disturbance alone, and hence, the formulation above is a DP problem.

\subsection{Bellman Equation}
The Bellman equation provides a sufficient condition that helps us verify that a given average reward is indeed optimal.

\begin{theorem}[\cite{Arapos_etal_93}, Thm.~6.2]  
\label{thm:Bellman}
    If $\rho \in \mathbb{R}$ and a bounded function $h:\mathscr{Z}\rightarrow \mathbb{R}$ satisfies $\forall z\in \mathscr{Z}$,
    \begin{equation*}
        \rho + h(z) = \sup_{u\in \mathscr{U}}\left[g(z,u) + \int P_W(dw|z,u)h(F(z,u,w))\right],
    \end{equation*}
    then $\rho^{*}=\rho$.
\end{theorem}

\section{$(d,k)$-RIC Channels} \label{sec:examples}

In this section, we study the application of Theorem \ref{Thm4} to certain input-constrained DMCs. Specifically, we impose $(d,k)$-RLL constraints, defined below, on input sequences.

\begin{definition}
	A binary sequence $\mathbf{x}=(x_1,x_2,\ldots)\in \{0,1\}^{*}$ is said to obey the $(d,k)$-RLL constraint, (for $0\leq d<k\leq \infty$) if each run of $0$s in $\mathbf{x}$ has length at most $k$, and any pair of successive $1$s is separated by at least $d$ $0$s.
	\label{def:RLL}
\end{definition}


It is easily verified that $(d,k)$-RIC DMCs are input-driven. Indeed, we take the state space $\mathscr{S}$ to be $\{0,1,2,\ldots,d\}$ if $k = \infty$, and $\{0,1,2,\ldots,k\}$ if $k < \infty$. The state transitions are shown in the edge-labelled directed graphs in Figs.~\ref{fig2} and \ref{fig3}: an edge $s \stackrel{x}{\longrightarrow} s'$ represents the transition $s' = f(s,x)$.
\begin{figure}[htbp]
  \centering
  \includegraphics[width=0.3\textwidth]{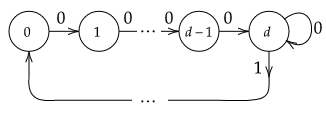}
   \caption{State transitions for the $(d,\infty)$-RLL constraint.}
   \label{fig2}
\end{figure}

\begin{figure}[htbp]
  \centering
  \includegraphics[width=0.4\textwidth]{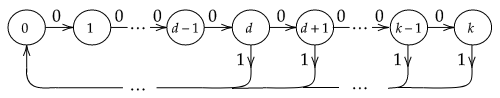}
   \caption{State transitions for the $(d,k)$-RLL constraint, $k < \infty$.}
   \label{fig3}
\end{figure}

Our assumption that the encoder and decoder share knowledge of the initial state $s_0$ is easily realized in this context, as they can \emph{a priori} agree upon a choice of $s_0$, e.g., $s_0 = 0$.

\begin{figure*}[t]
  \centering
  \includegraphics[width=0.5\textwidth]{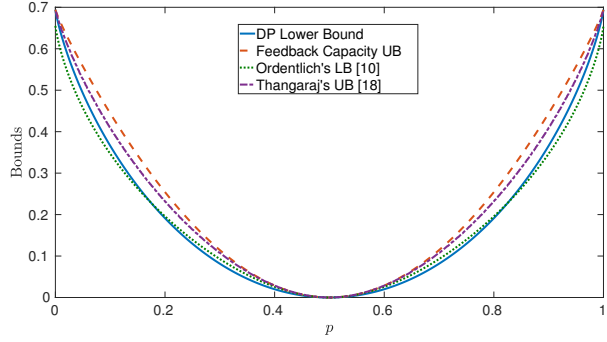}
   \caption{Comparison of the DP lower bound for the $(1,\infty)$-RIC BSC($p$) with bounds in \cite{Ordentlich}, \cite{NK_BSC} and \cite{Thangaraj}.}
   \label{fig4}
\end{figure*}

\begin{theorem}\label{BSC}
	The capacity of the $(d,\infty)$-RIC binary symmetric channel with cross-over probability $p$ obeys
	\begin{equation*}
	C_{d,\infty}^{\text{BSC}(p)}\geq \max_{a\in [0,1]} \frac{h_b(ap+\bar{a}\bar{p})-h_b(p)}{ad+1}.
	\end{equation*}
\end{theorem}
\begin{proof}
The DP state, $\mathbf{z}$, is a probability vector on $\mathscr{S} = \{0,1,2,\ldots,d\}$, with elements indexed as $z_s,$ $s\in \mathscr{S}$. As the channel is input-driven, we have $z_s \in \{0,1\}\, \forall s\in \mathscr{S}$, and exactly one $z_s$ can be equal to $1$. With some abuse of notation, we write the DP state as $z=i$ when $z_i=1$, for $0\leq i\leq d$.
	
	The disturbance $w$, is equal to $x$, and the action, $u$, is the stochastic matrix
	\renewcommand{\kbldelim}{[}
	\renewcommand{\kbrdelim}{]}
	
\[u\text{ }=	
\begin{blockarray}{ccc}
& 0 & 1 \\
\begin{block}{c[cc]}
  0 & 1 & 0 \\
  1 & 1 & 0 \\
  \vdots & \vdots & \vdots \\
  d-1 & 1 & 0 \\
  d & 1-a & a \\
\end{block}
\end{blockarray}
 \]
	
	where the rows correspond to the channel states, and $a\in [0,1]$. The next DP state is given by:
	\begin{equation*}
	F(z,u,x)=\psi(z,x),
	\end{equation*}
	where $\psi(z,x)=z^\prime$, if the edge $(z,z^\prime)$ labelled by $x$ exists in the presentation.
	From the conditional distribution in \eqref{eq1}, the reward function can be computed as:
	\begin{align*}
	g(z,u)&=H(Y|z,u)-H(Y|X)\\
	&=h_b\left(z_d(a+p-2ap)+p\sum\limits_{i=0}^{d-1}z_i\right)-h_b(p).
	\end{align*}
	Solving the Bellman equation of Theorem~\ref{thm:Bellman} entails identifying a scalar $\rho_p$ such that for a function $h_p:\mathscr{Z}\rightarrow \mathbb{R}$,
	\begin{multline}\label{eq2}
	\rho_p + h_p(z)=\max_{a\in [0,1]} [g(z,a)+(1-az_d)h_p(\psi(z,0))\\+az_dh_p(\psi(z,1))
	],
	\end{multline}
	for each $z\in \mathscr{Z}$. The set of $d+1$ equations in \eqref{eq2} can be split as:
	\begin{equation*}
	    \rho_p + h_p(i) = h_b(p)-h_b(p)+h_p(i+1)=h_p(i+1),
	\end{equation*}
	for $0\leq i\leq d-1$, since $z_d=0$, and
	\begin{equation}\label{BSC:eq1}
	    \rho_p = \max_a \left[h_b(a+p-2ap)-h_b(p)+a(h_p(0)-h_p(d))\right].
	\end{equation}
	From the first set of $d$ equations, we arrive at the fact that $d\rho_p=h_p(d)-h_p(0)$. Substituting this in equation \eqref{BSC:eq1}, we note that:
	\begin{equation}\label{eq3}
		\max_{a\in [0,1]}\left[h_b(ap+\bar{a}\bar{p})-h_b(p)-(ad+1)\rho_p\right]=0.
	\end{equation}
	Clearly, the choice
	\begin{equation*}
	\rho_p = \max_{a\in [0,1]}\frac{h_b(ap+\bar{a}\bar{p})-h_b(p)}{ad+1}
	\end{equation*}
	satisfies \eqref{eq3}.
\end{proof}

For $d=1$, figure \ref{fig4} shows plots of our DP lower bound, alongside the lower bound of Ordentlich \cite{Ordentlich}. Upper bounds on the capacity in the form of the feedback capacity of the $(1,\infty)$-RIC BSC($p$) \cite{NK_BSC}, and the dual capacity upper bound of Thangaraj \cite{Thangaraj} are also shown. Numerical evaluations indicate that the DP lower bound is close to the asymptotic bounds in \cite{Han1} as $p\rightarrow 0$, and in \cite{Ordentlich} as $p\rightarrow 0.5$, and extends these results to all values of $p$. Plots of the DP lower bound for $d=1,2,3$, are given in figure \ref{fig7}, with the unconstrained ($d=0$) capacity also indicated.

\begin{figure*}[t]
  \centering
  \includegraphics[width=0.5\textwidth]{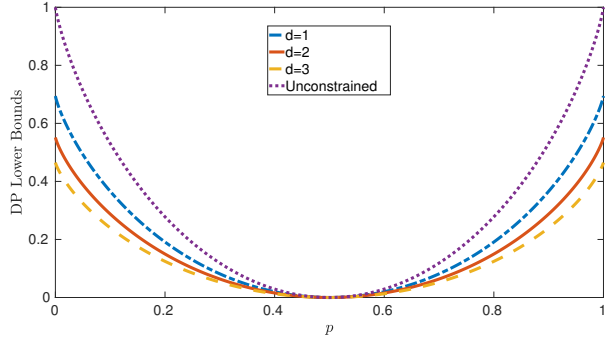}
   \caption{DP lower bounds for the $(1,\infty)$, $(2,\infty)$, $(3,\infty)$-RIC BSC($p$).}
   \label{fig7}
\end{figure*}

We now consider the $(d,k)$-RIC BSC($p$), with $k < \infty$.

\begin{theorem}\label{BSC2}
	The capacity of the $(d,k)$-RIC $(k<\infty)$ binary symmetric channel with cross-over probability $p$ satisfies
	\begin{equation*}
	C_{d,k}^{\text{BSC}(p)}\geq \max_{a_d,\ldots,a_{k-1}}\frac{\sum\limits_{i=d}^{k-1}\bigl(h_b(a_ip+\bar{a}_i\bar{p})-h_b(p)\bigr)\prod\limits_{j=d}^{i-1}(1-a_j)}{d+1+\sum\limits_{i=d}^{k-1}\prod\limits_{j=d}^{i}(1-a_j)},
    \end{equation*}
    where the maximization is over $a_d,\ldots,a_{k-1}\in [0,1]$, and an empty product is, by convention, equal to $1$.
\end{theorem}
\begin{proof}
This time, the DP state, $\mathbf{z}$, is a probability vector on $\mathscr{S} = \{0,1,2,\ldots,k\}$, with elements indexed as $z_s,$ $s\in \mathscr{S}$. As in the proof of Theorem~\ref{BSC}, we have $z_s \in \{0,1\}\, \forall s\in \mathscr{S}$, and exactly one $z_s$ can be equal to $1$, and so again, we write the DP state as $z=i$ when $z_i=1$, for $0\leq i\leq k$.
    
    The disturbance, $w$, is equal to $x$, and the action, $u$, is the stochastic matrix
    \[u\text{ }=	
\begin{blockarray}{ccc}
& 0 & 1 \\
\begin{block}{c[cc]}
  0 & 1 & 0 \\
  1 & 1 & 0 \\
  \vdots & \vdots & \vdots \\
  d-1 & 1 & 0 \\
  d & 1-a_d & a_d \\
  \vdots & \vdots & \vdots \\
  k-1 & 1-a_{k-1} & a_{k-1}\\
  k & 0 & 1\\
\end{block}
\end{blockarray}
 \]
	where $a_d,\ldots,a_{k-1}\in [0,1]$. The next DP state is given by $F(z,u,x)=\psi(z,x)$, where $\psi(z,x)=z^\prime$, if the edge $(z,z^\prime)$ labelled by $x$ exists in the presentation. The reward function then is:
	\begin{equation*}
	g(z,u)=h_b(p\delta + \bar{p}\bar{\delta})-h_b(p),
	\end{equation*}
	where $\delta := \sum\limits_{i=0}^{d-1}z_i + \sum\limits_{i=d}^{k-1}z_i(1-a_i)$. The Bellman equation in Theorem~\ref{thm:Bellman} simplifies to
	\begin{equation}
	\label{eq5}
	\rho_p + h_p(i) = h_p(i+1)
	\end{equation}
	for $0\leq i\leq d-1$, 
	\begin{multline}\label{eq:Bellman1}
	\rho_p + h_p(i)=\max_{a_i} [h_b(p\bar{a}_i+\bar{p}a_i)+(1-a_i)h_p(i+1)\\+a_ih_p(0)],
	\end{multline}
	for $d\leq i\leq k-1$, and 
	\begin{align}\label{eq:Bellman2}
	\rho_p =h_p(0)-h_p(k),
	\end{align}
	for $\rho_p\in \mathbb{R}$, and $h_p:\mathscr{Z}\rightarrow \mathbb{R}$. 
	
	Adding together the set of $d$ equations in \eqref{eq5} yields
	\begin{equation}\label{BSC:eq3}
	    d\rho_p = h_p(d)-h_p(0).
	\end{equation}
	Now, since we have that
	\begin{multline*}
	    \rho_p + h_p(k-1) = \max_{a_{k-1}}[h_b(p\bar{a}_{k-1}+\bar{p}a_{k-1})-h_b(p)\\+(1-a_{k-1})h_p(k)+a_{k-1}h_p(0)],
	\end{multline*}
	we substitute $h_p(k)$ as $h_p(0)-\rho_p$ from \eqref{eq:Bellman2}, giving
	\begin{multline}\label{BSC:inter2}
	    h_p(k-1)-h_p(0)=\max_{a_{k-1}}[h_b(p\bar{a}_{k-1}+\bar{p}a_{k-1})-h_b(p)\\+(a_{k-1}-2)\rho_p].
	\end{multline}
	We now substitute $h_p(k-1)$ from \eqref{BSC:inter2} in the penultimate equation of \eqref{eq:Bellman1}, to get $h_p(k-2)-h_p(0)$ in terms of $\rho_p$ and $p$ alone. Proceeding similarly, we arrive at:
	\begin{multline}\label{BSC:inter}
	    h_p(d)-h_p(0)=\max_{a_d,\ldots,a_{k-1}}[h_b(p\bar{a}_d+a_d\bar{p})-h_b(p)+\\\sum\limits_{i=d+1}^{k-1}(h_b(a_ip+\bar{a_i}\bar{p})-h_b(p))\prod\limits_{j=d}^{i-1}(1-a_j)\\-\rho_p(1+\sum\limits_{i=d}^{k-1}\prod\limits_{j=d}^{i}(1-a_j))].
	\end{multline}
	Using \eqref{BSC:eq3}, it is clear that the choice
	\begin{equation*}
	\rho_p = \max_{a_d,\ldots,a_{k-1}}\frac{\sum\limits_{i=d}^{k-1}(h_b(a_ip+\bar{a_i}\bar{p})-h_b(p))\prod\limits_{j=d}^{i-1}(1-a_j)}{d+1+\sum\limits_{i=d}^{k-1}\prod\limits_{j=d}^{i}(1-a_j)},
    \end{equation*}
    satisfies \eqref{BSC:inter}, where the empty product is taken to be $1$.
\end{proof}

We note here that the lower bound in the theorem is exactly equal to that presented in Lemma 5 of \cite{ZW88}, which evaluates a lower bound on the maximum mutual information rate among stationary Markovian input distributions on the graph in Fig. 3. Fig.~\ref{fig5} shows plots of the DP lower bound for the $(0,k)$-RIC BSC($p$), for $k=1,2,3$, alongside the capacity of the unconstrained ($k\rightarrow \infty$) BSC($p$).

\begin{figure*}
  \center
  \includegraphics[width=0.5\textwidth]{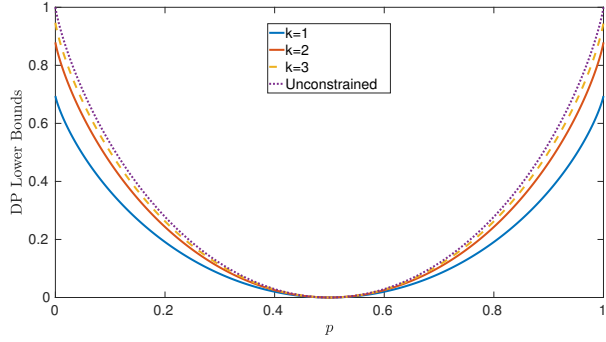}
   \caption{DP lower bounds for the $(0,1)$, $(0,2)$ and $(0,3)$-RIC BSC($p$).}
   \label{fig5}
\end{figure*}

We now move on to the input-constrained BEC($\epsilon$).
\begin{theorem}\label{BEC}
	The capacity of the $(d,\infty)$-RIC binary erasure channel with erasure probability $\epsilon$ satisfies
	\begin{equation*}
		C_{d,\infty}^{\text{BEC}(\epsilon)}\geq C_{d,\infty}\cdot \bar{\epsilon},
	\end{equation*}
	where $C_{d,\infty}=\underset{a\in [0,1]}{\max} \frac{h_b(a)}{ad+1}$ is the (noiseless) capacity of the $(d,\infty)$-RLL constraint.
\end{theorem}
\begin{proof}
    Just as in the proof of Theorem \ref{BSC}, the DP state, $\mathbf{z}$, is a probability vector on $\mathscr{S} = \{0,1,\ldots,d\}$, with elements indexed as $z_s\in \{0,1\}$, $s\in \mathscr{S}$, and we write the DP state, $z=i$, when $z_i=1$, for $0\leq i\leq d$.
    
    The disturbance, $w$, is equal to $x$, and the action, $u$, is the same stochastic matrix as in the proof of Theorem \ref{BSC}. Also, the next DP state is decided by the deterministic function, $\psi(\cdot)$, as defined in the proof of Theorem \ref{BSC}.
    
    Now, the reward function can be computed to be:
    \begin{align*}
	g(z,u)&=H(Y|z,u)-H(Y|X)\\
	&= H_{ter}(az_d\bar{\epsilon},\epsilon)-h_b(\epsilon)\\
	&\stackrel{(a)}{=} \bar{\epsilon}h_b(az_d), 
	\end{align*}
	where (a) follows from the fact that $H_{ter}(c\bar{d},d)=h_b(d)+\bar{d}h_b(c)$. Solving the Bellman equation entails identifying a scalar $\rho_\epsilon$ such that for a function $h_\epsilon:\mathscr{Z}\rightarrow \mathbb{R}$,
	\begin{multline}\label{BEC1}
	\rho_\epsilon + h_\epsilon(z)=\max_{a\in [0,1]} [g(z,a)+(1-az_d)h_\epsilon(\psi(z,0))\\+az_dh_\epsilon(\psi(z,1))
	],
	\end{multline}
	for each $z\in \mathscr{Z}$. The set of $d+1$ equations in \eqref{BEC1} can be split as:
	\begin{equation*}
	    \rho_\epsilon + h_\epsilon(i) = h_\epsilon(i+1),
	\end{equation*}
	for $0\leq i\leq d-1$, since $z_d=0$, and
	\begin{equation}\label{BEC3}
	    \rho_\epsilon = \max_a \left[\bar{\epsilon}h_b(a)+a(h_\epsilon(0)-h_\epsilon(d))\right].
	\end{equation}
	From the first set of $d$ equations, we arrive at $d\rho_\epsilon=h_\epsilon(d)-h_\epsilon(0)$. Substituting this in equation \eqref{BEC3}, we get:
	\begin{equation}\label{BEC4}
	    \max_{a\in [0,1]}[\bar{\epsilon}h_b(a)-\rho_\epsilon(ad+1)]=0.
	\end{equation}
	The choice
	\begin{equation*}
	    \rho_\epsilon = \bar{\epsilon}\max_{a}\left\{\frac{h_b(a)}{ad+1}\right\}
	\end{equation*}
	satisfies \eqref{BEC4}. We note that at $\epsilon=0$, $Y^N$ is equal to $X^N$. Hence, the coefficient of $\bar{\epsilon}$ is equal to the noiseless capacity of the $(d,\infty)$-RIC, by the remark following Theorem \ref{Thm3}.
\end{proof}
For $d=1$, a comparison between the DP lower bound and the ``memory $1$'' dual capacity upper bound of Thangaraj \cite{Thangaraj} are shown in figure \ref{fig6}, along with a plot of the feedback capacity \cite{NK}. Numerical evaluations indicate that the DP lower bound also closely approximates the asymptotic ($\epsilon\rightarrow 0$) lower bound of Li et al. \cite{Han2}. We now provide a lower bound on the capacity of the $(d,k)$-RIC BEC($\epsilon$). 

\begin{figure*} 
  \centering
  \includegraphics[width=0.5\textwidth]{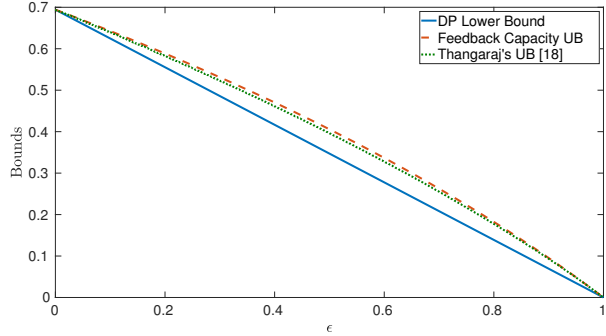}
   \caption{Comparison of the DP lower bound for the $(1,\infty)$-RIC BEC($\epsilon$) with bounds in \cite{NK} and \cite{Thangaraj}.}
   \label{fig6}
\end{figure*}
\begin{theorem}\label{BEC2}
	The capacity of the $(d,k)$-RIC $(k<\infty)$ binary erasure channel with erasure probability $\epsilon$ satisfies
	\begin{equation*}
		C_{d,k}^{\text{BEC}(\epsilon)}\geq C_{d,k} \cdot \bar{\epsilon},
	\end{equation*}
	where $C_{d,k}=\underset{a_d,\ldots,a_{k-1}}{\max} \frac{\sum\limits_{i=d}^{k-1}h_b(a_i)\prod\limits_{j=d}^{i-1}(1-a_j)}{d+1+\sum\limits_{i=d}^{k-1}\prod\limits_{j=d}^{i}(1-a_j)}$ is the (noiseless) capacity of the $(d,k)$-RLL constraint.
\end{theorem}
\begin{proof}
    The DP state, $\mathbf{z}$, is a probability vector on $\mathscr{S} = \{0,1,\ldots, k\}$, with elements indexed as $z_s\in \{0,1\}$, $s\in \mathscr{S}$, and we write the DP state, $z=i$, when $z_i=1$, for $0\leq i\leq k$.
    
    The disturbance, $w$, is equal to $x$, and the action, $u$, is the same stochastic matrix as in the proof of Theorem \ref{BSC2}. The next DP state is dictated by the deterministic function, $\psi(\cdot)$, defined in the proof of Theorem \ref{BSC}. The reward function is computed to be:
    \begin{align*}
	g(z,u)&=H(Y|z,u)-H(Y|X)\\
	&= H_{ter}\left(\bar{\epsilon}\left(z_k+\sum\limits_{i=d}^{k-1}z_ia_i\right),\epsilon\right)-h_b(\epsilon)\\
	&= \bar{\epsilon}h_b\left(z_k+\sum\limits_{i=d}^{k-1}z_ia_i\right), 
	\end{align*}
	where $H_{ter}(\cdot,\cdot)$ denotes the ternary entropy function, and the last equality follows from the property of the ternary entropy function. The Bellman equations can then be split as:
	\begin{align}\label{BEC5}
	\rho_\epsilon + h_\epsilon(i)&=\max_{a_d,\ldots,a_{k-1}} [g(z,u)+(1-z_k-\sum\limits_{i=d}^{k-1}a_iz_i)h_\epsilon(\psi(z,0))]\\
	&=h_\epsilon(i+1),
	\end{align}
	for $0\leq i\leq d-1$, and
	\begin{align}\label{BEC6}
	\rho_\epsilon + h_\epsilon(i)=\max_{a_i} [\bar{\epsilon}h_b(a_i)+(1-a_i)h_\epsilon(i+1)+a_ih_\epsilon(0)],
	\end{align}
	for $d\leq i\leq k-1$, and 
	\begin{align}\label{BEC7}
	\rho_\epsilon =h_\epsilon(0)-h_\epsilon(k),
	\end{align}
	for $\rho_\epsilon\in \mathbb{R}$, and $h_\epsilon:\mathscr{Z}\rightarrow \mathbb{R}$. Again, the first set of $d$ equations gives us
	\begin{equation}\label{BEC8}
	    d\rho_\epsilon = h_\epsilon(d)-h_\epsilon(0).
	\end{equation}
	Now, since we have that
	\begin{multline*}
	    \rho_\epsilon + h_\epsilon(k-1) = \max_{a_{k-1}}[\bar{\epsilon}h_b(a_{k-1})+(1-a_{k-1})h_\epsilon(k)+\\a_{k-1}h_\epsilon(0)],
	\end{multline*}
	we substitute \eqref{BEC8}, giving
	\begin{multline*}
	    \rho_\epsilon + h_\epsilon(k-1)=\max_{a_{k-1}}[\bar{\epsilon}h_b(a_{k-1})+h_\epsilon(0)-(1-a_{k-1})\rho_\epsilon],
	\end{multline*}
	and hence,
	\begin{multline*}
	    h_\epsilon(k-1)-h_\epsilon(0)=\max_{a_{k-1}}[\bar{\epsilon}h_b(a_{k-1})+(a_{k-1}-2)\rho_\epsilon].
	\end{multline*}
	Proceeding similarly, we arrive at
	\begin{multline}
	    h_\epsilon(d)-h_\epsilon(0)=\max_{a_d,\ldots,a_{k-1}}[\bar{\epsilon}h_b(a_d)+\bar{\epsilon}\sum\limits_{i=d+1}^{k-1}h_b(a_{i})\prod\limits_{j=d}^{i-1}(1-a_j)\\-\rho_\epsilon(1+\sum\limits_{i=d}^{k-1}\prod\limits_{j=d}^{i}(1-a_j))].
	\end{multline}
	Now, using \eqref{BEC8}, we see that $\rho_\epsilon$ obeys
	\begin{equation*}
	\rho_\epsilon = \bar{\epsilon}\cdot \max_{a_d,\ldots,a_{k-1}}\left(\frac{\sum\limits_{i=d}^{k-1}h_b(a_i)\prod\limits_{j=d}^{i-1}(1-a_j)}{d+1+\sum\limits_{i=d}^{k-1}\prod\limits_{j=d}^{i}(1-a_j)}\right),
    \end{equation*}
    where the empty product is taken to be $1$. Again, the coefficient of $\bar{\epsilon}$ is equal to the noiseless capacity of the $(d,k)$-RIC, by the remark following Theorem \ref{Thm3}.
\end{proof}
\begin{remark}
    In the theorems above, the fact that the coefficients of $\bar{\epsilon}$ are indeed equal to the noiseless capacities also follows from the observation that the coefficients are the entropies of the maxentropic Markov chains on the graphs in Figs.~\ref{fig2} and \ref{fig3}, and hence, by Theorem 3.23 of \cite{Roth}, are equal to the noiseless capacities.
\end{remark}

\section{Proofs}\label{sec:proofs}

\begin{proof}[Proof of Theorem \ref{Thm3}]
By way of \eqref{LB}, we have 
	\begin{align}
	I(X^N;Y^N \mid s_0) 
	& \ge I(Y^N \to X^N \mid s_0) \notag \\
	&\geq \sum_{t=1}^N I(X_t;Y_t \mid X^{t-1},s_0) \label{ineq:thm3.1_proof}\\
	&= \sum_{t=1}^N I(X_t,S_{t-1};Y_t \mid X^{t-1},s_0),  \notag
	\end{align}	
	the last equality following from the state evolution of input-driven FSCs. Hence, via Theorem~\ref{thm:FSC_cap}, we have 
	\begin{align*}
	C &= \lim_{N\rightarrow \infty} \max_{\{P(x_t|x^{t-1},s_0)\}_{t=1}^{N}} \frac{1}{N} I(X^N;Y^N \mid s_0) \\
	&\ge \lim_{N\rightarrow \infty} \max_{\{P(x_t|x^{t-1},s_0)\}_{t=1}^{N}} \frac{1}{N} \sum_{t=1}^N I(X_t,S_{t-1};Y_t \mid X^{t-1},s_0)\\
	&= \sup_{\{P(x_t|x^{t-1})\}_{t\geq 1}} \liminf_{N\rightarrow \infty} \frac{1}{N} \sum_{t=1}^N I(X_t,S_{t-1};Y_t \mid X^{t-1}),
	\end{align*}
	the last equality above following by the arguments in Lemma 4 of \cite{Weiss}, the conditioning on $s_0$ being suppressed in the notation. Finally, we can replace the supremum over $\{P(x_t|x^{t-1})\}_{t \ge 1}$ by a supremum over input distributions of the form $\{P(x_t|s_{t-1})\}_{t \ge 1}$, possibly at the expense of another inequality.
\end{proof}

\begin{proof}[Proof of Theorem~\ref{Thm5}]
For any fixed $s_0$, we have
	\begin{align}
	C_N &:= \max_{\{P(x_t|x^{t-1},s_0)\}_{t=1}^{N}} \frac{1}{N}I(X^{N};Y^{N} \mid s_0) \notag \\
	&\stackrel{(a)}{\geq} \max_{\{P(x_t|x^{t-1},s_0)\}_{t=1}^{N}} \frac{1}{N}\sum_{t=1}^{N} I(X_t;Y_t \mid X^{t-1},s_0) \notag \\
	&\stackrel{(b)}{=} \max_{\{P(x_t|x^{t-1},s_0)\}_{t=1}^{N}} \frac{1}{N} \sum_{t=1}^{N} I(X_t;Y_t \mid S_{t-1},X^{t-1},s_0), \label{ineq1:thm3.3_proof}
	\end{align}
where (a) is by \eqref{ineq:thm3.1_proof}, and (b) is by the fact that for an input-driven channel, $S_{t-1}$ is determined by $X^{t-1}$ and $s_0$.

Now, let $\{Q(x|s,v)\} \in \mathscr{Q}$ be an aperiodic input distribution, so that $\mathscr{G}(\{Q(x|s,v)\})$ has a single closed communicating class, ${\mathscr{G}}_0$, that is also aperiodic. By the connectedness of the FSC, arguing as in Lemma~1 of \cite{Single}, there is some vertex $v_0$ of the $\mathscr{V}$-graph such that $(s_0,v_0)$ is in ${\mathscr{G}}_0$. 

We will continue the chain of inequalities from \eqref{ineq1:thm3.3_proof} by specifying an input distribution $\{P(x_t|x^{t-1},s_0)\}_{t=1}^N$ in terms of $\{Q(x|s,v)\}$. For all $t \ge 1$, set
$$
P(x_t \mid x^{t-1},s_0) = Q(x_t \mid s_{t-1},v_{t-1}),
$$
where $s_{t-1} = f(s_0,x^{t-1})$ is the state at time $t-1$ reached by the FSC starting at $s_0$ and driven by the inputs $x^{t-1}$, and $v_{t-1} = \Phi(v_0,x^{t-1})$ is the vertex of the $\mathscr{V}$-graph at the end of the path labelled by $x^{t-1}$ starting at $v_0$. Note that this choice of input distribution induces the following Markov chain:
	\begin{equation}\label{MC}
	X^{t-1}\textbf{\textemdash}(S_{t-1},V_{t-1})\textbf{\textemdash}(X_t,Y_t),
	\end{equation}
where $S_{t-1} = f(s_0,X^{t-1})$ and $V_{t-1} = \Phi(v_0,X^{t-1})$.

Thus, carrying on from \eqref{ineq1:thm3.3_proof}, with the input distribution $\{P(x_t|x^{t-1},s_0)\}_{t=1}^{N}$ specified as above, we have
	\begin{align} 
	C_N &\ge \frac{1}{N} \sum_{t=1}^N I(X_t;Y_t \mid S_{t-1},X^{t-1},s_0) \notag \\
	&= \frac{1}{N} \sum_{t=1}^N I(X_t;Y_t \mid S_{t-1},V_{t-1}). \label{CN_ineq}
	\end{align}	
The last equality is due to the fact that $I(X_t;Y_t \mid S_{t-1},X^{t-1},s_0) = I(X_t;Y_t \mid S_{t-1},V_{t-1},X^{t-1},s_0)$, since $V_{t-1}$ is determined by $X^{t-1}$ and the (fixed) vertex $v_0$, and the latter mutual information equals $I(X_t;Y_t \mid S_{t-1},V_{t-1})$ by the Markov chain in \eqref{MC}. Finally, taking limits as $N \to \infty$ in \eqref{CN_ineq}, we obtain the desired bound $C \ge I_Q(X;Y|S,V)$, since 
$$
\lim_{N \to \infty} \frac{1}{N} \sum_{t=1}^N I(X_t;Y_t \mid S_{t-1},V_{t-1}) = I_Q(X;Y \mid S,V)
$$
by the ergodicity of the Markov chain on $\mathscr{S} \times \mathscr{V}$ induced by the aperiodic input distribution $\{Q(x|s,v)\}$.
\end{proof}

\section{Conclusions and Future Work} \label{sec:conc}
In this work, novel lower bounds on the capacities of input-driven FSCs were derived. The main idea was the lower bounding of the mutual information rate by the reverse directed information rate. A DP formulation of the lower bound was given, which was then applied to input-constrained memoryless channels, resulting in simple analytical expressions that are valid for all values of the channel parameters, thereby extending known asymptotic results. Furthermore, an alternative single-letter lower bound on the capacity was derived, using the concept of directed $\mathscr{V}$-graphs.

This paper focuses on lower bounding the reverse directed information rate only. The lower bounds on capacity that this approach yields can be improved by estimating or bounding the forward directed information rate from the input distribution that optimizes the reverse rate.

\section*{Acknowledgment}
This work was supported in part by a MATRICS grant (no.\ MTR/2017/000368) administered by the Science and Engineering Research Board (SERB), Govt.\ of India.

\bibliographystyle{IEEEtran}
{\footnotesize
\bibliography{references}}

\end{document}